\documentclass[a4paper]{article}
\usepackage{amsmath}
\usepackage{amssymb}
\usepackage{amsthm}
\usepackage{dsfont}
 \newtheorem{thm}{Theorem}
 \newtheorem{prop}{Proposition}
 \newtheorem{lemma}{Lemma}
 \newcommand{\bbr}{{\mathbb R}}
\newcommand{\bbs}{{\mathbb S}}
\usepackage{authblk}

\def\scri{\mathcal{I}^+}
\def\be{\begin{equation}}
\def\ee{\end{equation}}
\newcommand{\bn}{{\bf{n}}}
\newcommand{\bp}{{\bf{p}}}
\newcommand{\bq}{{\bf{q}}}

\begin{document}
\title{The massless Einstein-Boltzmann system with a conformal gauge singularity in an FLRW background}

\author[1]{Ho Lee\footnote{holee@khu.ac.kr}}
\author[2]{Ernesto Nungesser\footnote{ernesto.nungesser@icmat.es}}
\author[3]{Paul Tod\footnote{tod@maths.ox.ac.uk}}

\affil[1]{Department of Mathematics and Research Institute for Basic Science, Kyung Hee University, Seoul, 02447, Republic of Korea}
\affil[2]{Instituto de Ciencias Matem\'{a}ticas (CSIC-UAM-UC3M-UCM), 28049 Madrid, Spain}
\affil[3]{Mathematical Institute, University of Oxford, Oxford OX2 6GG}

\maketitle

\begin{abstract}
We obtain finite-time existence for the massless Boltzmann equation, with a range of soft cross-sections, in an FLRW background with data given at the initial singularity. In the case of positive cosmological constant we obtain long-time existence in proper-time for small data as a corollary.
\end{abstract}

\section{Introduction}
The work in \cite{A} showed that there exists a well-posed Cauchy problem for cosmological solutions of the physical massless Einstein--Vlasov equations, which is to say the Einstein equations with massless, collisionless matter as source, with an initial \emph{conformal gauge singularity}, \cite{LT}. A conformal gauge singularity, or isotropic singularity, is essentially a curvature singularity which can be removed by conformally rescaling the metric, and which is therefore one at which the Weyl tensor is finite. There are reasons for thinking that the initial singularity of our actual physical universe is one at which the Weyl tensor is finite, \cite{rp2}, and that it might in fact be a conformal gauge singularity, \cite{CCC}. With a conformal gauge singularity, in the unphysical, rescaled space-time the singularity is a smooth space-like hypersurface, which we'll call the \emph{bang surface}, and data may be given at this hypersurface. From the stand-point of the physical space-time, there is a curvature singularity at the bang surface, so that these points are not part of space-time, but the singularity can be regularised by rescaling the metric and then data is given actually at what was the singularity. The question naturally arises of extending the work of \cite{A} to the Einstein--Boltzmann equations, that is to say the Einstein equations with collisional matter as source. 

Mathematically, with any of these matter models, the problem is to find an extended set of conformal Einstein equations in the conformally extended manifold for which finite-time existence can be proved with data at the bang surface. This was achieved for a range of polytropic perfect fluids in \cite{AT99a}, for massless Einstein--Vlasov with a spatially homogeneous metric in \cite{AT99}, and for massless Einstein--Vlasov without symmetry in \cite{A}. For both kinds of source, the conformal Einstein equations can be formulated as a Fuchsian system with the pole located at the bang surface but the appropriate Cauchy data are strikingly different in the two cases: for the perfect fluid case the data are simply a Riemannian 3-metric, in fact the metric of the bang surface, with no separate degrees of freedom for the fluid, while for the Einstein--Vlasov case there is a single datum, the initial distribution function subject only to non-negativity and a vanishing dipole condition. The initial metric is extracted from the initial distribution function.

The question was raised\footnote{K.~Anguige and A.~R.~Rendall; Private communication to P.~Tod, 2000.} as to whether the Einstein--Boltzmann case, which is to say the Einstein-Vlasov case with the inclusion of a collision term in the Vlasov equation, might make a bridge between these two cases, with perfect fluids as one limit and Vlasov as the other, depending on the scattering cross-section. This possibility was explored informally in \cite{T03} and received some support. The intention now is to proceed more rigorously. The problem is difficult because the collision term inevitably has singularities which have to be dealt with (these are visible in equations  (\ref{Boltzmann1}) and (\ref{bolt2}) below).

One considers massless particles because, near the bang, one expects the particles to be so energetic that their rest-mass, even if nonzero, would have negligible effect. Likewise, near the bang the cosmological constant $\Lambda$ would be expected to be physically irrelevant, and there have been studies \cite{T07} to indicate that indeed, in the cases so far studied, the inclusion of nonzero rest-mass or $\Lambda$ have negligible effect, but we include the case of a positive cosmological constant to make contact with Penrose's conformal cyclic cosmology (or CCC, \cite{CCC}). In that theory the rest-mass of all elementary particles is zero near the bang and in the remote future. 

There are rather few mathematical results on the Einstein--Boltzmann system since the original existence results of \cite{BCB}, and fewer still on massless Einstein-Boltzmann, where the masslessness introduces extra poles in the collision term (see (\ref{bolt3}): in the case of nonzero mass $m$ the term $q$ in the denominator is replaced by $(q^2+m^2R^2)^{1/2}$ which is bounded away from zero except at the bang surface). This is to be contrasted with the collisionless case, where the linearity of the Vlasov equation means that, in situations with symmetry, many explicit solutions can be simply written down, and in systems without symmetry one still has linearity. 

As a first step in extending \cite{A}, we consider the homogeneous and isotropic case. The metric is Friedman-Lema\^itre-Robertson-Walker (hereafter FLRW) and the scale-factor is fixed by imposing the Einstein equations with trace-free energy-momentum tensor as source. We assume the source is determined by a homogeneous and isotropic distribution function that is subject to the Boltzmann equation, so that this is a self-consistent Einstein-Boltzmann solution, although the Einstein and Boltzmann equations decouple. Since there are rather few mathematical results on the massless Einstein-Boltzmann system, our strategy with respect to an appropriate choice of cross-section is pragmatic: we look among the families of physically-reasonable cross-sections classified for example in \cite{DEJ} for some which lead eventually to a tractable initial-value problem in our setting. We need to make assumptions about the scattering cross-section, specifically about its behaviour at large energies and the dependence on the FLRW scale-factor that this implies.  This is to be expected when one wishes to give data at the bang surface, where the scale factor vanishes. We also need to regularise the Boltzmann equation by redefining the time-coordinate (see equations (\ref{s1}) and (\ref{bolt5})) -- proper time is defined by the physical metric and this is singular at the bang surface, so that one should be prepared to change the time-coordinate to regularise the equations -- and this also imposes a restriction on the scattering cross-sections considered. However all these constraints are satisfied by the family of cross-sections that we consider in our main result, Theorem 1. It is worth noting that the origin of the redefined time coordinate $s$ in (\ref{s1}) can be chosen to coincide with the origin of proper time $t$, namely at the bang surface, but with a positive $\Lambda$ there is only a finite amount of $s$-time before the conformal infinity $\scri$ which is attached at infinite proper time $t$. Since we eventually prove existence for finite $s$-time, one may reduce the norm of the initial distribution function $f_0$ so that the solution exists for finite $s$-time but infinite $t$-time in the case $\Lambda>0$.

Once this FLRW case can has been successfully handled, the next step is to extend the analysis to spatially-homogeneous space-times and a homogeneous and isotropic distribution function, following \cite{AT99} and \cite{LN2}. 

\medskip

The plan of the paper is as follows. In Section 2, we discuss the flat FLRW metric and the Einstein-Boltzmann equations with the corresponding symmetry assumption. The massless Boltzmann equation gives rise to a trace-free energy-momentum tensor and, given a choice of positive, zero or negative cosmological constant,  this determines the FLRW scale-factor and therefore the whole metric via the Einstein equations. We have still to prove well-posed-ness of the Boltzmann equation with some choice of cross-section, but this is now a problem set in a given space-time metric. We give a  discussion of possible scattering cross-sections, following the discussions and classifications in \cite{DEJ}, \cite{LR} and \cite{S10}, and we focus on a small family of so-called soft ones which satisfy the constraints we need to impose. To add the bang surface as a boundary requires conformally rescaling the metric, so we also describe the behaviour under conformal rescaling of the massless Boltzmann equation in an FLRW background. 

In Section 3, we give the statement and proof of long-time existence in $s$-time of solutions of a modified Boltzmann equation, with a cut-off imposed at small momentum to remove the singularity in the collision integral at zero centre-of-mass energy -- this singularity arises because we are treating massless particles, and this is a characteristic difficulty of the massless case, not present in the massive case. This is the content of Proposition 1 in Section 3.1. Then, in Section 3.2, comes the key novelty and hardest part of the proof which is to prove convergence of the solution to a solution of the Boltzmann equation as the cut-off is removed. This only goes through for a finite $s$-time, i.e. existence is proved for finite time but in the redefined time coordinate. This is Theorem 1. To be explicit:

\begin{thm}
Let $f_0\in L^1(\bbr^3)$ be initial data for the distribution function at $t=0$, equivalently $s=0$, satisfying $f_0\geq 0$ and 
\[
\int_{\bbr^3}f_0(p)p^m d^3\!p<\infty
\] 
for $-2\leq m\leq 2$. Then, there exists a positive value $T$ of the redefined time coordinate $s$ such that the massless Boltzmann equation \eqref{bolt5} with scattering cross-sections in the family given in (\ref{cx1}) has a unique solution $f\in C^1([0,T];L^1(\bbr^3))$ for the distribution function in $s$-time which is non-negative and satisfies
\[
\sup_{0\leq s\leq T}\int_{\bbr^3} f(s,p)p^m d^3\!p\leq C_T,
\]
for $-1\leq m\leq 2$.
\end{thm}
Since we have the FLRW scale-factor already, by solving the Einstein equations, we therefore obtain a solution of the coupled Einstein-Boltzmann equations with positive, zero or negative cosmological constant, at least for a finite amount of proper time. However, in the $\Lambda>0$ case, as noted above, we shall see that the whole infinite range in proper time $t$ corresponds to a finite amount of $s$-time, and by adjusting the size of the initial distribution function $f_0$ we can ensure that the solution exists for all proper time. The theorem also provides a continuation criterion: the solution can be extended in the time coordinate $s$ as long as the five moments corresponding to integer $m$ in the range $-2\leq m\leq 2$ remain finite.

\section{Einstein-Boltzmann system in the FLRW case with spatially flat topology}\label{s2}
We'll suppose for simplicity that the spatial curvature is zero, and then the spatially-flat FLRW models have the following metric:
\begin{align}\label{metric}
g=-dt^2+ R^2 (dx^2 + dy^2+dz^2),
\end{align}
where $R=R(t)>0$ is the scale factor, vanishing at the initial singularity which we'll locate at $t=0$. The Einstein equations reduce to equations for $R$, and, for a trace-free energy-momentum tensor corresponding to a radiation fluid source (so $p=\rho/3$) and positive cosmological constant, they are given by
\begin{align}
\frac{{\dot{R}}^{2}}{R^{2}}&=\frac{8 \pi}{3}\rho+\frac{\Lambda}{3},\label{Einstein1}\\
\frac{\ddot{R}}{R}&=-\frac{8\pi}{3} \rho+\frac{\Lambda}{3},\label{Einstein2}
\end{align}
with the overdot for $d/dt$, and where the energy density $\rho$ will be given in terms of the distribution function $f$ introduced below, and $\Lambda$ is the  cosmological constant, which we'll usually assume to be non-negative.

\subsection{The scale factor and conformal time}\label{ss2.1}
In the massless case, we can find the scale factor at the outset, so that the Einstein equations are solved and only the Boltzmann equation remains for consideration.

Adding the Einstein equations \eqref{Einstein1}--\eqref{Einstein2} we have
\begin{align*}
\frac{{\dot{R}}^{2}}{R^{2}}+\frac{\ddot{R}}{R}=\frac{2\Lambda}{3},
\end{align*}
which can be expressed as
\begin{align*}
\frac{d^2}{dt^2} (R^2)= \frac{4\Lambda}{3}R^2.
\end{align*}
We consider a solution which starts with an initial singularity at $t=0$. In the case of vanishing cosmological constant ($\Lambda=0$) we have for $t\geq 0$
\begin{align}\label{scale1}
R=C_1 t^{\frac12},
\end{align}
while for a positive cosmological constant ($\Lambda>0$) we set $\Lambda=3H^2$ and obtain
\begin{align}\label{scale2}
R=C_2 \sqrt{\sinh (2Ht)},
\end{align}
where $C_1$ and $C_2$ are some positive constants of integration related to the matter content\footnote{One could also consider a negative cosmological constant, $\Lambda=-3H^2$ when $R=C_3\sqrt{\sin(2HT)}$ and the metric recollapses. Since we eventually obtain a finite-time existence result, the solution in this case could last until recollapse, but this universe model is less attractive than the others.}. Note that $\dot{R}/R\rightarrow H$ as $t\rightarrow\infty$ so that this $H$ is the limiting Hubble parameter.

It will be convenient to introduce \emph{conformal time} $\tau$ by 
\begin{align*}
\frac{dt}{d\tau}= R,
\end{align*}
with the origins of $t$ and $\tau$ coinciding: this will be important -- the singularity at $t=0$ is also at $\tau=0$, and this will be where we wish to give data.

Then, the metric \eqref{metric} can be written as
\[
g=R^2(-d\tau^2+dx^2 + dy^2+dz^2)= R^2 \eta,
\]
where $\eta$ is the Minkowski metric\footnote{And not to be confused with conformal time, which for us is denoted by $\tau$.}. In the absence of a cosmological constant the relation between the usual cosmic time $t$ and the conformal time $\tau$ is given by
\begin{align}
\tau=\frac{2}{C_1} t^{\frac12},\label{conformal time}
\end{align}
where $C_1$ is the constant given in \eqref{scale1}. Hence when $\Lambda=0$, up to a constant the scale factor equals the conformal time, and both time coordinates, $t$ and $\tau$, go to infinity together. On the other hand, in the presence of a positive cosmological constant, i.e. with $\Lambda>0$, (\ref{conformal time}) holds approximately near the bang but there is only a finite amount of conformal time between the bang and future infinity (or $\scri$) at $t=\infty$. One calculates
\begin{align*}
\tau=  \int^t_0  \frac{ \sqrt{2} e^{Ht'} dt'}{ C_2 \sqrt{e^{4Ht'}-1}}= \frac{\sqrt{2}}{C_2H} \int_1^{e^{Ht}} \frac{dx}{\sqrt{x^4-1}} \leq \frac{\sqrt{2}}{C_2H} \int^{\infty}_1 \frac{dx}{\sqrt{x^4-1}}.
\end{align*}
The final expression gives the total amount of conformal time between the initial singularity and future infinity at $t=\infty$. Note that the integral is
\begin{align*}
\int^{\infty}_1 \frac{dx}{\sqrt{x^4-1}}= \frac{ \sqrt{\pi} \Gamma (\frac54)}{ \Gamma(\frac43)} \approx 1.31103.
\end{align*}

\subsection{The massless Boltzmann equation}\label{ss2.2}

For massless particles the momentum $p_\alpha$ as a one-form satisfies
\[
p_\alpha p_{\beta}g^{\alpha\beta} = 0,
\]
where the momentum $p_\alpha=(p_0, p_1, p_2,p_3)$ is defined by
\[
p_\alpha dx^\alpha=p_0dt+p_1dx+p_2dy+p_3dz
\]
(so these are coordinate indices) and is assumed to be future directed, i.e., $p^0>0$. It will often be convenient to write
\[{\bf{p}}:=(p_1,p_2,p_3),\]
and we denote for simplicity
\begin{equation}\label{pdef}
p:=|{\bf{p}}|=(\sum_{i=1}^3(p_i)^2)^{1/2},
\end{equation}
then it's easy to see that $p$ and $\bp$ are constant along geodesics, and for null momenta we have
\be\label{n1}
p^0=-p_0=R^{-1}p.
\ee
For consistency with the metric, we assume that the distribution function is also homogeneous and isotropic, so that it takes the form
\[
f=f(t,p).
\]

With this distribution function, the energy density $\rho$ is defined by
\begin{align}
\rho=\int_{\bbr^3} (p_0)^2f\omega_p=4\pi R^{-4} \int_{\bbr^+} f p^3\,  dp,\label{energy dentisy}
\end{align}
where $\omega_p=dp_1dp_2dp_3/(p^0\sqrt{-g})=d^3p/(pR^2)$ is the volume-form on the null cone for $p_\alpha$ (when integrating a function isotropic in momenta we have of course $d^3p=4\pi p^2dp$). The Boltzmann equation \eqref{Boltzmann1} is coupled to the Einstein equations \eqref{Einstein1}--\eqref{Einstein2} through the energy density \eqref{energy dentisy}. The last vestige of the Einstein equations is the relation between the constant $C_1$ in (\ref{scale1}) (or $C_2$ in (\ref{scale2}) if $\Lambda>0$) and the conserved quantity in (\ref{energy dentisy}). From (\ref{Einstein1}):
\[ \int_{\bbr^3} pfd^3\!p= R^4\rho=\frac{3}{32\pi}C_1^4\mbox{  or  }\frac{3H^2}{8\pi}C_2^4\]
which tie together the scale factor and $f$ via the energy density in the cases $\Lambda=0$ and $\Lambda>0$ respectively.

\medskip

The number current vector is in general
\[N_\alpha=\int_{\bbr^3} p_\alpha  f\omega_p,\]
and the only non-zero component, given isotropy, is
\begin{align}
N^0= 4\pi R^{-3} \int_{\bbr^+} f p^2\, dp.
\end{align}
The divergence conditions on the energy-momentum tensor and the number current vector show that the following are constants of the motion:
\begin{equation}\label{c1}R^3 N^0= 4\pi\int_{\bbr^+} p^2f\,dp,\;\;\;R^4\rho= 4\pi\int_{\bbr^+}p^3f\,dp,\end{equation}
which we shall want to be finite, which in turn imposes conditions on initial data for $f$.

We consider only binary collisions, in which particles with null momenta $p_\alpha,q_\alpha$ collide to produce particles with null momenta $p'_\alpha,q'_\alpha$, all four future-pointing. Conservation of momentum is assumed,
and then the post-collision momenta $p'_\alpha,q'_\alpha$ can be represented in terms of the pre-collision momenta $p_\alpha,q_\alpha$ and a unit 3-vector $\omega_i$ (there are different ways to achieve this -- the details of our parametrisation are in section \ref{s3}). We write $d\omega$ for the standard volume form on the unit 2-sphere, thought of as the sphere of $(\omega_i)$.

The Boltzmann equation (see e.g. \cite{LR}) for homogeneous and isotropic $f$ in an FLRW background reduces to
\begin{align}
\partial_t f = R^{-3} \int_{\bbr^3}\int_{\bbs^2}v_M \sigma(h,\omega)\Big(f(p')f(q')-f(p)f(q)\Big)d\omega d^3\!q.\label{Boltzmann1}
\end{align}
Here, the integral operator on the right hand side is the \emph{collision operator} and describes the effect of binary collisions between particles. With $p_\alpha,q_\alpha$ as the pre-collision momenta, the relative momentum $h$ and the total energy $s$ are defined by
\[
h=\sqrt{(p_\alpha-q_\alpha)(p^\alpha-q^\alpha)},\quad
s=-(p_\alpha+q_\alpha)(p^\alpha+q^\alpha),
\]
so that for massless particles $h^2=s=-2p_\alpha q^\alpha$, and the M{\o}ller velocity $v_M$ is defined by
\[
v_M=\frac{h\sqrt{s}}{p^0q^0}.
\]
For massless particles then (\ref{Boltzmann1}) simplifies to
\begin{equation}\label{bolt2}
\partial_t f = R^{-3} \int_{\bbr^3}\int_{\bbs^2}\frac{1}{p^0q^0}s\sigma(h,\omega)\Big(f(p')f(q')-f(p)f(q)\Big)d\omega d^3\!q.
\end{equation}

The quantity $\sigma$ in (\ref{Boltzmann1}) is called the scattering cross-section or scattering kernel and depends only on $h$ and the scattering angle $\Theta$, which in turn is defined by
\[
\cos\Theta=(p'_\alpha-q'_\alpha)(p^\alpha-q^\alpha)/h^2.
\]
One of the extra complications of the massless case can be seen by noting that in the massive case one has
\[p^0=(m^2+R^{-2}p^2)^{1/2}\]
for mass $m$, in place of (\ref{n1}) and at least away from the space-time singularity at $R=0$, this is bounded away from zero. Therefore in the massless case the right-hand-side of (\ref{bolt2}) has singularities at $p=0$ or $q=0$ which are absent in the massive case away from $R=0$ (or in Minkowski space where $R\equiv 1$).

The results we obtain depend critically on the behaviour of the scattering cross-section at large energy $s$ so we next discuss some possibilities.
\subsection{Scattering cross-sections}\label{ss2.3}
We recall the discussion of hard and soft scattering cross-sections in \cite{DEJ}, \cite{LR} and \cite{S10}, modified slightly for the massless case ($g$ in \cite{LR} is $h$ here, and for massless particles $s=g^2$). 
\begin{itemize}
 \item 
A cross-section is \emph{soft} if there exists a real $\gamma>-2$ and a real $b$ in the range $0<b<\mbox{min}(4,\gamma+4)$ and positive constants $c_1,c_2,c_3$ such that
\[c_1h^{-b}\sigma_0(\Theta)\leq\sigma(h,\Theta)\leq c_2h^{-b}\sigma_0(\Theta),\]
with
\[\sigma_0(\Theta)\leq c_3(\sin\Theta)^\gamma.\]
\item
A cross-section is \emph{hard} if there exists a real $\gamma>-2$, a real $a$ in the range $0\leq a\leq\gamma+2$ and a real $b$ in the range $0<b<\mbox{min}(4,\gamma+4)$ and positive constants $c_1,c_2,c_3$ such that

\[c_1h^a\sigma_0(\Theta)\leq \sigma(h,\Theta)\leq c_2(h^a+h^{-b})\sigma_0(\Theta),\]
with
\[\sigma_0(\Theta)\leq c_3(\sin\Theta)^\gamma.\]

\medskip

Note that soft cross-sections all decrease with increasing $h$ while hard cross-sections grow unless $a=0$.

\item A particular case of a hard cross-section is the so-called \emph{hard spheres} cross-section for which
\[\sigma(h,\Theta)=\sigma_0=\mbox{  constant},\]
and $a=0$. This cross-section was used in \cite{baz1} with a new time coordinate $d\tilde{t}=dt/R^3$ which we'll discuss below, in section \ref{ss2.5}.
\item 
The cross-section for {\emph{Israel particles}} was introduced in \cite{is} as a mathematically tractable relativistic counterpart of classical, massive Maxwell particles. Transformed to the massless case and with our definition\footnote{Note that $\sigma$ in \cite{is} is $h\sigma$ in our terms.} of $\sigma$ it is
\[\sigma(h,\Theta)=\Gamma(\Theta)/h^3,\]
with no restriction on $\Gamma$. If $\Gamma$ is constant then this is classified as soft with $b=3$.
%
%
\end{itemize}
\subsection{Equilibrium solutions}
\label{ss2.4}
A distribution function of the form
\[f(t,p)=\mbox{exp}(-\alpha-\beta p),\]
with real constants $\alpha,\beta$ automatically gives zero on the right in (\ref{Boltzmann1}) by conservation of momentum and so automatically solves (\ref{Boltzmann1}), regardless of the choice of cross-section. Furthermore, as is familiar, these are the only solutions giving zero on the right. They are the \emph{equilibrium solutions} and $\alpha,\beta$ can be related to the constants of integration found in (\ref{c1}):
\[R^3N^0=8\pi e^{-\alpha}\beta^{-3},\;\;\;R^4\rho=24\pi e^{-\alpha}\beta^{-4}.\]
A distribution function of this form is perfectly well-defined at $t=0$ and in $t>0$ it can be written
\[f(t,p)=e^{-\alpha-\beta R(t)p^0},\]
when it is recognisably a Maxwell state with temperature $kT=1/(\beta R(t))$ which diverges initially and then decays to zero.

Evidently, given values for the constants of integration $R^3N^0$ and $R^4\rho$ there is a unique corresponding Maxwell state, and it is a question of interest whether other solutions converge to one of these states.

\subsection{Conformally rescaling the Boltzmann equation}\label{ss2.5}
As a general rule, conformal rescaling (see e.g. \cite{pr}) is the change of metric
\[g_{\alpha\beta}\rightarrow\hat{g}_{\alpha\beta}=\Omega^2g_{\alpha\beta},\]
with smooth $\Omega$. This is accompanied by
\[g^{\alpha\beta}\rightarrow\hat{g}^{\alpha\beta}=\Omega^{-2}g^{\alpha\beta},\]
and corresponding changes to the connection and curvature (see e.g. \cite{pr}). In the context of massless kinetic theory one has
\[\hat{p}_\alpha=p_\alpha,\;\;\hat{f}(x^\alpha,p_\beta)=f(x^\alpha,p_\beta).\;\;\hat{\omega}_p=\Omega^{-2}\omega_p,\;\;d\hat{\omega}=d\omega.\]
As consequences of these we also have
\[\hat{p}=p,\;\;\hat{s}=\Omega^{-2}s,\;\;\hat{h}=\Omega^{-1}h,\;\;\hat{p}^0=\Omega^{-2}p^0,\]
and so on. 

With the specific metric (\ref{metric}) of interest in this paper of course one doesn't need much of the theory of conformal rescaling: with $\Omega=R^{-1}$ we have $\hat{g}_{ab}=\eta_{ab}$, the flat metric and keeping track of powers of $\Omega$ corresponds to keeping track of powers of $R$. The Boltzmann equation (\ref{bolt2}) (not rescaled but with powers of $R$ introduced and simplified) can be written
\begin{equation}\label{bolt3} 
 \partial_t f = R^{-1} \int_{\bbr^3}\int_{\bbs^2}\frac{s\sigma}{pq}\Big(f(p')f(q')-f(p)f(q)\Big)d\omega d^3\!q,
\end{equation}
or in conformal time
\begin{equation}\label{bolt4} \partial_\tau f =  \int_{\bbr^3}\int_{\bbs^2}\frac{s\sigma}{pq}\Big(f(p')f(q')-f(p)f(q)\Big)d\omega d^3\!q.
 \end{equation}
We want to find a well-posed Cauchy problem with data at $t=0=\tau$ and clearly the possibility of this is now tied to the behaviour of $s\sigma$ for small $R$. We have
\[h^2=s=R^{-2}(pq-{\bf{p}}\cdot{\bf{q}}),\]
so that $h$ diverges as $O(R^{-1})$ near $t=0$. Now
\begin{itemize}
 \item For a {\emph{hard}} cross-section from section \ref{ss2.3} we have
\[s\sigma=O(h^{a+2})=O(R^{-2-a}),\]
and $a\geq 0$. The right-hand-side in (\ref{bolt4}) is $O(R^{-2-a})=O(\tau^{-2-a})$ near the bang surface and so diverges at least quadratically at $\tau=0$: we cannot expect solutions to exist for wide classes of data. 
\item For a {\emph{soft}} cross-section from section \ref{ss2.3}, the first term in the integrand in (\ref{bolt4}) is
\[\frac{s\sigma}{pq}\sim \frac{h^{2-b}}{pq}=\frac{(pq-{\bf{p}}\cdot{\bf{q}})^{(2-b)/2}}{R^{2-b}pq}.\]
If $b> 2$ then this has a pole at $h=0$, which happens when the incoming particles have parallel momenta (or when one of them is zero, but there are already singularities at $pq=0$). If $b\leq 2$ then there is a pole at $R=0$ which we can seek to absorb in a redefinition of the time variable: introduce $s$ with\footnote{This $s$ is not to be confused with total energy $s$ which can always be eliminated as $h^2$.}
\be\label{sdef}ds=R^{b-2}d\tau=R^{b-3}dt,\ee
then
\[\partial_s f=R^{2-b} \partial_\tau f\]
and the factor $R^{b-2}$ is cancelled from (\ref{bolt4}). Near $t=0$, and using either (\ref{scale1}) or (\ref{scale2}) in (\ref{sdef}) gives
\[ds\sim t^{(b-3)/2}dt\mbox{   and  }s \sim t^{(b-1)/2}.\]
We need the singularity, which is at $t=0$, to be at $s=0$ (or at least at finite $s$) so we require $b>1$. Thus to avoid the pole at $h=0$ and to ensure the data surface is at $s=0$ we should restrict to soft cross-sections with $b$ in the range $1<b<2$ (the possibility $b=2$ is ruled out in Section 3).

\item Note that for the hard spheres cross-section $\sigma=\sigma_0=$constant, (\ref{bolt3}) becomes
\[\partial_t f = R^{-3} \int_{\bbr^3}\int_{\bbs^2}\frac{\sigma_0(pq-{\bf{p}}\cdot{\bf{q}})}{pq}\Big(f(p')f(q')-f(p)f(q)\Big)d\omega d^3\!q.\]
In \cite{baz1}, a remarkable explicit solution was given for this Boltzmann equation in an FLRW background. These authors redefine the time coordinate by
\[ds=R^{-3}dt\mbox{  so that  }s\sim C-2t^{-1/2}\mbox{  and  }\partial_sf=R^3\partial_tf,\]
which removes the factor $R^{-3}$ from the Boltzmann equation, which they then solve by an ingenious method. However this change of time coordinate pushes the initial singularity off to minus infinity in $s$. This is not a cause of concern in \cite{baz1}, where the authors evolve only into the future from a positive value of $t$, but it doesn't serve our purposes. In fact one can see that the distribution function in equation (22) of \cite{baz1} becomes negative in the past (since their ${\mathcal{K}}(\tau)$ becomes negative in the past).

\item For Israel particles, $s\sigma=\Gamma/h=O(R)$ which goes to zero at $\tau=0$. Thus the right-hand-side in (\ref{bolt4}) goes to zero, and even in terms of proper time, in (\ref{bolt3}), all powers of $R$ cancel and the Boltzmann equation is \emph{exactly} as it would be in flat space. None-the-less the pole at $h=0$ makes this case currently unmanagable.

\end{itemize}

\section{Solutions to the Boltzmann equation}\label{s3}
Now we seek to prove well-posedness of (\ref{bolt4}) with data at the initial singularity and, motivated by section \ref{ss2.3}, the family of scattering cross-sections
\be\label{cx1}\sigma(h,\Theta)=C_1h^{-b}\ee
with positive real $C_1$, not to be confused with $C_1$ in section \ref{ss2.1}, and real $b$ in the range $1<b<2$. After some general theory, in Section 3.1 we modify the Boltzmann equation by imposing a cut-off at small momentum, in order to remove a pole in the collision integral. We prove long-time existence of a positive distribution function for this modified Boltzmann equation in Section 3.1. Then in Section 3.2 we prove that this solution to the modified Boltzmann equation converges to a positive solution of the actual (unmodified) Boltzmann equation at least for a finite interval in $s$. This is the content of Theorem 1 and our main result.

We work in the rescaled picture, so that the unphysical null momentum is parametrised as
\be\label{cx2}\hat{p}^\alpha=(p,{\bf{p}})\mbox{  with  }p=\left(\Sigma_{i=1}^3(p_i)^2\right)^{1/2}.\ee
It will be convenient and will aid readability to omit all hats from now on but the reader is reminded that these momenta are unphysical and should be hatted.

The physical centre of mass energy is as before
\[h^2=2R^{-2}(pq-{\bf{p}}\cdot{\bf{q}})\]
but it is convenient to introduce
\[\varrho=Rh=\sqrt{2(pq-{\bf{p}}\cdot{\bf{q}})},\]
(again, strictly speaking $\varrho=\hat{h}$ but the current definition avoids the proliferation of hats and simplifies equations).
As noted, all collisions will be binary conserving four-momentum which as usual is expressed as
\[p_\alpha+q_\alpha=p'_\alpha+q'_\alpha,\]
with the out-momentum $p'$ parametrised in terms of a unit vector $\omega^i$ by
\begin{align}
\left(
\begin{array}{c}
p'\\
{\bf{p}}'
\end{array}
\right)=
\left(
 \begin{array}{c}
 \displaystyle
 \frac12(p+q)+\frac12({\bf{p}}+{\bf{q}})\cdot{\bf{\omega}}\\
  \displaystyle
  \frac12({\bf{p}}+{\bf{q}})
   +\frac12\varrho{\bf{\omega}} 
+\frac12(({\bf{p}}+{\bf{q}})\cdot{\bf{\omega}})\frac{({\bf{p}}+{\bf{q}})}{(p+q+\varrho)}
\\
\end{array}
\right),\label{p'2}
\end{align}
and similar for $q'$ with the sign switched on $\omega^i$. Note this is the same expression as for this problem in Minkowski case -- all appearances of the scale-factor have been eliminated from these expressions. 

For the Boltzmann equation with the cross-sections considered we have
\[ \partial_\tau f =  \int_{\bbr^3}\int_{\bbs^2}\frac{C_1\varrho^{2-b}}{R^{2-b}pq}\Big(f(p')f(q')-f(p)f(q)\Big)d\omega d^3q,\]
so we redefine the time-coordinate again, introducing $s$ by
\be\label{s1}ds=R^{b-2}d\tau=R^{b-3}dt.\ee
From the discussion around equation (\ref{sdef}) in section \ref{ss2.5} we know that, with the assumed restrictions on $b$, we can choose $s=0$ at $t=0$, so that the data-surface is fixed. It should be noted that, as we are free to choose $R$ as one or other of the choices (\ref{scale1}) or (\ref{scale2}) we simultaneously treat the cases $\Lambda=0$ and $\Lambda>0$. One important difference between the cases is that for $\Lambda=0$ there is an infinite amount of $s$-time in the future, but for $\Lambda>0$ there is only a finite amount of $s$-time before infinity in $t$. Since we eventually obtain finite-time existence in $s$, this could be sufficient for infinite time in $t$.

Now the Boltzmann equation becomes
\begin{equation}\label{bolt5} \partial_s f = C_1 \int_{\bbr^3}\int_{\bbs^2}\frac{\varrho^{2-b}}{pq}\Big(f(p')f(q')-f(p)f(q)\Big)d\omega d^3q.
 \end{equation}

Below, we show that the equation \eqref{bolt5} with \eqref{p'2} admits unique local-in-time solutions. We first collect some lemmas:

\begin{lemma}\label{Lem.varrho}
Let $p^\alpha$, $q^\alpha$, $p'^\alpha$, and $q'^\alpha$ be pre- and post-collision (unphysical) momenta given by (\ref{cx2}) and \eqref{p'2}. Then, the following holds:
\[
\varrho^2\leq 4\min (pq, p'q').
\]
\end{lemma}
\begin{proof}
We have
\[
\varrho^2 =  2(pq -{\bf{p}}\cdot{\bf{q}})\leq 4pq
\]
Since $\varrho$ is a collision invariant, the inequality also holds for $p'q'$.
\end{proof}

\begin{lemma}\label{Lem.varrho2}
The relative momentum $\varrho$ can be written as
\[
\varrho^2 = 4pq\sin^2\frac{\theta}{2},
\]
where $\theta$ is the angle between the three-dimensional vectors ${\bf{p}}$ and ${\bf{q}}$.
\end{lemma}
\begin{proof}
This is a simple calculation. Note that
\begin{align*}
\varrho^2  = 2(p q - {\bf{p}}\cdot{\bf{q}}) = 2pq(1-\cos\theta) = 4 pq \sin^2\frac{\theta}{2}.
\end{align*}
\end{proof}

The following is a special case of the Povzner inequality. We refer to \cite{LR,SY} for more general versions of the inequality in the case of massive particles. 
\begin{lemma}\label{Lem.Povzner}
Let $p^\alpha$, $q^\alpha$, $p'^\alpha$, and $q'^\alpha$ be pre- and post-collision (unphysical) momenta given by \eqref{p'2}. Then, the following holds:
\[
(p')^2 + (q')^2 - p^2 - q^2 \leq 2 pq.
\]
\end{lemma}
\begin{proof}
By \eqref{p'2} we have
\begin{align*}
&(p')^2 + (q')^2 - p^2 - q^2\\
& = \bigg(\frac{p+q}{2}+\frac{(p_j+q_j)\omega^j}{2}\bigg)^2 + \bigg(\frac{p+q}{2}-\frac{(p_j+q_j)\omega^j}{2}\bigg)^2 - p^2 - q^2\\
& = \frac{(p+q)^2}{2} + \frac{((p_j+q_j)\omega^j)^2}{2}  - p^2 - q^2\\
& \leq \frac{(p+q)^2}{2} + \frac{(p_j+q_j)(p^j+q^j)}{2}  - p^2 - q^2\\
& = p q + p_j q^j\\
&\leq 2pq,
\end{align*}
where we used $p^2 = p_j p^j$ and $q^2 = q_j q^j$.
\end{proof}

\begin{lemma}\label{Lem.Povzner1}
Let $p^\alpha$, $q^\alpha$, $p'^\alpha$, and $q'^\alpha$ be pre- and post-collision (unphysical) momenta given by \eqref{p'2}. Then, for any $0<a<1$, there exists $C>0$ so that
\[
\int_{\bbs^2}\frac{1}{p'}d\omega = \int_{\bbs^2}\frac{1}{q'}d\omega \leq \frac{C}{\varrho^a (p+q)^{1-a}}.
\]
\end{lemma}
\begin{proof}
For simplicity we write $n_\alpha = p_\alpha + q_\alpha$ such that
\[
n^0 = p + q,\quad {\bf{n}}\cdot\omega = (p_j + q_j)\omega^j,\quad |{\bf{n}}| = \sqrt{\sum_{j=1}^3(p_j + q_j)^2}.
\]
Then, we have
\begin{align*}
\int_{\bbs^2}\frac{1}{p'}d\omega & = \int_{\bbs^2} \frac{2}{n^0+ \bn\cdot\omega} d\omega\\
& = 4\pi\int_0^\pi \frac{\sin\theta}{n^0+|\bn|\cos\theta} d\theta\allowdisplaybreaks\\
& = \frac{4\pi}{|\bn|}\ln \bigg(\frac{n^0+|\bn|}{n^0 - |\bn|}\bigg)\\
& = \frac{8\pi}{|\bn|}\ln\bigg(\frac{n^0+|\bn|}{\varrho}\bigg),
\end{align*}
where we used $(n_0)^2 -|\bn|^2 = \varrho^2$. Since $n_0 = \sqrt{\varrho^2 +|\bn|^2}$, we obtain
\begin{align*}
\frac{8\pi}{|\bn|}\ln\bigg(\frac{n^0+|\bn|}{\varrho}\bigg) & = \frac{8\pi}{|\bn|}\ln\bigg(\frac{\sqrt{\varrho^2 +|\bn|^2}+|\bn|}{\varrho}\bigg)\\
& = \frac{8\pi}{\varrho}\frac{\ln\Big(\frac{|\bn|}{\varrho} +\sqrt{1+\frac{|\bn|^2}{\varrho^2}}\Big)}{\frac{|\bn|}{\varrho}}\allowdisplaybreaks\\
& = \frac{8\pi}{\varrho^a(n^0)^{1-a}}\bigg(\frac{n^0}{\varrho}\bigg)^{1-a}\frac{\ln\Big(\frac{|\bn|}{\varrho} +\sqrt{1+\frac{|\bn|^2}{\varrho^2}}\Big)}{\frac{|\bn|}{\varrho}}\\
& = \frac{8\pi}{\varrho^a(n^0)^{1-a}}\bigg(1+\frac{|\bn|^2}{\varrho^2}\bigg)^{\frac{1-a}{2}}\frac{\ln\Big(\frac{|\bn|}{\varrho} +\sqrt{1+\frac{|\bn|^2}{\varrho^2}}\Big)}{\frac{|\bn|}{\varrho}}.
\end{align*}
Here, we note that the quantity
\[
(1+x^2)^{\frac{1-a}{2}}\frac{\ln(x + \sqrt{1+x^2})}{x}
\]
is bounded on $[0,\infty)$. Hence, we conclude that
\[
\frac{8\pi}{|\bn|}\ln\bigg(\frac{n^0+|\bn|}{\varrho}\bigg)\leq \frac{C}{\varrho^a(n^0)^{1-a}}.
\]
The calculation for $q'^0$ is the same, and this completes the proof. 
\end{proof}

\begin{lemma}\label{Lem.Povzner2}
Let $p^\alpha$, $q^\alpha$, $p'^\alpha$, and $q'^\alpha$ be pre- and post-collision (unphysical) momenta given by \eqref{p'2}. Then, the following holds:
\[
\int_{\bbs^2}\frac{1}{(p')^2}d\omega = \int_{\bbs^2}\frac{1}{(q')^2}d\omega = \frac{16\pi}{\varrho^2}.
\]
\end{lemma}
\begin{proof}
We use the notation $n^\alpha = p^\alpha + q^\alpha$ as in the previous lemma. A simple calculation shows that
\begin{align*}
\int_{\bbs^2}\frac{1}{(p')^2}d\omega & = \int_{\bbs^2}\frac{4}{(n^0 + \bn\cdot\omega)^2}d\omega\\
& = 8\pi\int_0^\pi \frac{\sin\theta}{(n^0 + |\bn|\cos\theta)^2} d\theta\\
& = \frac{8\pi}{|\bn|}\bigg(\frac{1}{n^0 - |\bn|} - \frac{1}{n^0 + |\bn|}\bigg)
 = \frac{16\pi}{\varrho^2}.
\end{align*}
The calculation for $q'$ is the same, and this completes the proof. 
\end{proof}

We now prove the existence of solutions to the equation \eqref{bolt5} with \eqref{p'2}. We first consider a modified Boltzmann equation to remove the singularity $(pq)^{-1}$. Here, we follow the results of \cite{ark,LR,SY}, where the authors considered massive particles. The arguments are standard (see Chapter 6 of \cite{CIP} for more details about the arguments), and we observe that the argument applies to the massless case with the modification we make. In the Section 3.2 we remove the modification and establish existence of positive solution for the unmodified Boltzmann equation.

\subsection{Solutions to the modified equation}\label{ss3.1}
We take $\varepsilon>0$ and consider the modified collision operator defined by
\[
Q_\varepsilon(f,f) = C_1\iint_{\varrho\geq\varepsilon}\frac{\varrho^{2-b}}{pq}\Big(f(p')f(q')-f(p)f(q)\Big)d\omega d^{3}\!q.
\]
Note that the quantity $\varrho^{2-b}/(pq)$ is bounded for $\varrho\geq\varepsilon$ with $1<b<2$ by Lemma \ref{Lem.varrho}. The modified operator $Q_\varepsilon$ is bounded in $L^1$ as follows:
\begin{align}
&\|Q_\varepsilon(f,f)\|_{L^1}\nonumber\\
&\leq C_1\iiint_{\varrho\geq\varepsilon}\frac{\varrho^{2-b}}{pq}|f(p')||f(q')| d\omega d^3\!q d^3\!p + C_\varepsilon\iiint_{\varrho\geq\varepsilon}|f(p)||f(q)| d\omega d^3\!q d^3\!p\nonumber\\
&\leq C_1\iiint_{\varrho\geq\varepsilon}\frac{\varrho^{2-b}}{p'q'}|f(p')||f(q') | d\omega d^3\!q' d^3\!p' + C_\varepsilon \|f\|_{L^1}^2\nonumber\\
&\leq C_\varepsilon \|f\|_{L^1}^2,\label{Qbdd}
\end{align} 
where we used Lemma \ref{Lem.varrho} and the standard change of variables
\begin{align}\label{dpdq}
\frac{1}{pq}d^3\!pd^3\!q = \frac{1}{p'q'}d^3\!p'd^3\!q'.
\end{align}
For $f,g\in L^1$ we consider the following expression
\begin{align}
&Q_\varepsilon(f,f)-Q_\varepsilon(g,g)\nonumber\\
&=\frac{C_1}{2}\iint_{\varrho\geq \varepsilon}\frac{\varrho^{2-b}}{pq}\Big((f+g)(p')(f-g)(q') + (f+g)(q')(f-g)(p')\nonumber\\
&\qquad\qquad\qquad - (f+g)(p)(f-g)(q) - (f+g)(q)(f-g)(p)\Big)d\omega d^3\!q.\label{Qsymm}
\end{align}
By the same estimates as in \eqref{Qbdd} we obtain
\begin{align}\label{QLip}
\|Q_\varepsilon(f,f)-Q_\varepsilon(g,g)\|_{L^1}
\leq C_\varepsilon \|f+g\|_{L^1}\|f-g\|_{L^1},
\end{align}
which shows that $Q_\varepsilon$ is Lipschitz continuous with respect to $f$. The estimates \eqref{Qbdd} and \eqref{QLip} show that the modified equation
\begin{align}\label{Bmod}
\partial_s f = Q_\varepsilon(f,f),\quad f(0)=f_0\geq 0,
\end{align}
admits a unique solution in $C^1([0,T];L^1)$ for some $T>0$, where $T$ depends only on $\varepsilon$ and $\|f_0\|_{L^1}$. 

Next, we show that the solutions of \eqref{Bmod} are non-negative for $f_0\geq 0$. Note that the solution satisfies
\begin{align}\label{intQ}
\int_{\bbr^3} Q_\varepsilon(f,f)d^3\!p = 0,
\end{align}
which is a consequence of \eqref{dpdq}. Hence, we obtain
\begin{align}\label{cons0}
\int_{\bbr^3} fd^3\!p = \int_{\bbr^3} f_0d^3\!p = \|f_0\|_{L^1}.
\end{align}
If $f$ is non-negative on $[0,T]$, then we have $\|f(T)\|_{L^1} = \|f_0\|_{L^1}$ and obtain existence on $[T,2T]$ by the same arguments again. In this way we obtain existence of solutions on any finite time interval, so that the solutions exist globally in time. In other words, the non-negativity of $f$ leads to the global-in-time existence. 

To prove the non-negativity of $f$ we rewrite the equation \eqref{Bmod} for a large $\mu>0$ as
\[
\partial_s f + \mu f \int_{\bbr^3} f d^3\!q = Q_\varepsilon(f,f)+\mu f \int_{\bbr^3} fd^3\!q,
\]
and introduce the following equation:
\begin{align}\label{Bmodh}
\partial_s h + \mu_0 h = \Gamma^\mu_\varepsilon(h),\quad h(0)=f_0\geq 0,
\end{align}
where $\mu_0 = \mu \|f_0\|_{L^1}$ and $\Gamma^\mu_\varepsilon$ is defined by
\begin{align*}
\Gamma^\mu_\varepsilon(h)=Q_\varepsilon(h,h)+\mu h \int_{\bbr^3} hd^3\!q.
\end{align*}
Note that the solutions of \eqref{Bmod} satisfy the equation \eqref{Bmodh} for any $\mu>0$ because of \eqref{cons0}. Since the solutions are unique, it is now enough to obtain the non-negativity of solutions to the equation \eqref{Bmodh}. We first notice that the operator $\Gamma^\mu_\varepsilon$ is non-negative, if $\mu$ is sufficiently large. Moreover, it is monotone as follows.
\begin{lemma}\label{Lem.monotone}
The operator $\Gamma^\mu_\varepsilon$ is monotone, i.e., for $f\geq g\geq 0$ in $L^1$ we have
\[
\Gamma^\mu_\varepsilon(f)\geq \Gamma^\mu_\varepsilon(g).
\]
\end{lemma}
\begin{proof}
Let $f\geq g\geq 0$. We need to show 
\[
Q_\varepsilon(f,f)-Q_\varepsilon(g,g) +\mu f\int_{\bbr^3} fd^3\!q - \mu g\int_{\bbr^3} gd^3\!q\geq 0.
\]
By the expression \eqref{Qsymm}, we have
\begin{align*}
&Q_\varepsilon(f,f)-Q_\varepsilon(g,g)\nonumber\\
&=\frac{C_1}{2}\iint_{\varrho\geq \varepsilon}\frac{\varrho^{2-b}}{pq}\Big((f+g)(p')(f-g)(q') + (f+g)(q')(f-g)(p')\Big)d\omega d^3\!q\nonumber\\
&\quad- \frac{C_1}{2}\iint_{\varrho\geq \varepsilon}\frac{\varrho^{2-b}}{pq}\Big((f+g)(p)(f-g)(q) + (f+g)(q)(f-g)(p)\Big)d\omega d^3\!q,
\end{align*}
and the first integral is non-negative by the assumption. The second integral is estimated as follows.
\begin{align*}
&-\frac{C_1}{2}\iint_{\varrho\geq \varepsilon}\frac{\varrho^{2-b}}{pq}\Big((f+g)(p)(f-g)(q) + (f+g)(q)(f-g)(p)\Big)d\omega d^3\!q\\
&\geq - C_\varepsilon \bigg((f+g)(p)\int_{\bbr^3} (f-g)(q) d^3\!q + (f-g)(p)\int_{\bbr^3} (f+g)(q) d^3\!q\bigg).
\end{align*}
The remaining terms can be written as
\begin{align*}
&\mu f\int_{\bbr^3} fd^3q -\mu g\int_{\bbr^3} gd^3\!q\\
& = \frac{\mu}{2}\bigg((f+g)(p)\int_{\bbr^3} (f-g)(q)d^3\!q  + (f-g)(p) \int_{\bbr^3} (f+g)(q)d^3\!q\bigg).
\end{align*}
We combine the above estimates to conclude that 
$\Gamma^\mu_\varepsilon(f)- \Gamma^\mu_\varepsilon(g) \geq 0$ for a sufficiently large $\mu>0$. 
\end{proof}
\noindent Consider an iteration: let $h_0 = 0$, and define
\[
h_{n+1}(s) = e^{- \mu_0 s} f_0 + \int_0^s e^{-\mu_0 (s - r)}\Gamma^\mu_\varepsilon (h_n) (r) dr.
\]
Note that $h_1 = e^{-\mu_0 s} f_0\geq 0 = h_0$. Suppose inductively that $h_n \geq h_{n-1}$ for some $n\geq 1$. Then, by Lemma \ref{Lem.monotone} we have
\[
h_{n+1}-h_n = \int_0^s e^{-\mu_0 (s - r)}\Big(\Gamma^\mu_\varepsilon (h_n) - \Gamma^\mu_\varepsilon (h_{n-1}) \Big)dr\geq 0.
\]
Hence, we obtain a monotonically increasing sequence
\[
0\leq h_1\leq h_2\leq  \cdots.
\]
Moreover, the sequence $\{ h_n\}$ is bounded in $L^1$. To see this, note that $h_0$ satisfies $\|h_0\|_{L^1}\leq \|f_0\|_{L^1}$, so suppose inductively that $\|h_{n}\|_{L^1}\leq \|f_0\|_{L^1}$ for some $n\geq 0$. From the iteration we have
\begin{align*}
\int_{\bbr^3} h_{n+1} d^3\!p = e^{-\mu_0 s}\int_{\bbr^3} f_0 d^3\!p + \int_0^s e^{-\mu_0(s-r)}\int_{\bbr^3} \Gamma^\mu_\varepsilon(h_n)d^3\!pdr.
\end{align*}
By the property \eqref{intQ} we observe that
\[
\int_{\bbr^3} \Gamma^\mu_\varepsilon(h_n)d^3\!p = \mu \bigg(\int_{\bbr^3} h_n d^3\!p\bigg)^2 \leq \mu \|f_0\|_{L^1}^2.
\]
Hence, we estimate
\begin{align*}
\int_{\bbr^3} h_{n+1} d^3\!p &\leq e^{-\mu_0 s}\| f_0 \|_{L^1}+\mu \|f_0\|^2_{L^1} \int_0^s e^{-\mu_0(s-r)}dr = \|f_0\|_{L^1},
\end{align*}
where we used $\mu_0 = \mu \|f_0\|_{L^1}$, and this proves  $\|h_n\|_{L^1}\leq \|f_0\|_{L^1}$ for all $n$. Consequently, the sequence $\{ h_n\}$ has a limit $h$, which is non-negative and bounded by $\|f_0\|_{L^1}$, by the Monotone Convergence Theorem. The limit $h$ solves the equation \eqref{Bmodh}, and therefore we conclude that the solution $f$ to the equation \eqref{Bmod} is non-negative. We therefore obtain the existence part of the following result:

\begin{prop}\label{Prop.mod}
Let $\varepsilon>0$ be given. For any initial data $f_0\in L^1(\bbr^3)$ with $f_0\geq 0$ the modified equation \eqref{Bmod} has a unique solution $f\in C^1([0,\infty);L^1(\bbr^3))$ which is non-negative. Moreover, if $f_0$ satisfies
\[
\int_{\bbr^3}f_0(p)p^md^3\!p<\infty
\] 
for $-2\leq m\leq 2$, then there exists $T>0$ such that
\[
\sup_{0\leq s\leq T}\int_{\bbr^3} f(s,p)p^m d^3\!p\leq C_T,
\]
where $C_T$ is independent of $\varepsilon$.
\end{prop}
\begin{proof}
Multiplying the equation \eqref{Bmod} by $p^m$ and integrating it over $\bp$ we obtain
\begin{align}\label{dsfm}
\frac{d}{ds} \int_{\bbr^3} f(p)p^md^3\!p = \int_{\bbr^3} Q_\varepsilon(f,f)(p)p^md^3\!p.
\end{align}
We use the property \eqref{dpdq} to write the right hand side as follows. For the gain term we observe
\begin{align*}
&C_1\iiint_{\varrho\geq\varepsilon} \frac{\varrho^{2-b}}{pq}f(p')f(q') p^md\omega d^3\!q d^3\!p\\
&= C_1\iiint_{\varrho\geq\varepsilon} \frac{\varrho^{2-b}}{p'q'}f(p')f(q') p^md\omega d^3\!q' d^3\!p'\allowdisplaybreaks\\
&= C_1\iiint_{\varrho\geq\varepsilon} \frac{\varrho^{2-b}}{pq}f(p)f(q) (p')^md\omega d^3\!q d^3\!p.
\end{align*}
By interchanging $p$ and $q$ we have
\begin{align*}
&C_1\iiint_{\varrho\geq\varepsilon} \frac{\varrho^{2-b}}{pq}f(p')f(q') p^md\omega d^3\!q d^3\!p\\
&= C_1\iiint_{\varrho\geq\varepsilon} \frac{\varrho^{2-b}}{pq}f(p)f(q) (q')^md\omega d^3\!q d^3\!p.
\end{align*}
Similarly, the loss term can be written as
\begin{align*}
&C_1\iiint_{\varrho\geq\varepsilon} \frac{\varrho^{2-b}}{pq}f(p)f(q) p^md\omega d^3\!q d^3\!p\\
&= C_1\iiint_{\varrho\geq\varepsilon} \frac{\varrho^{2-b}}{pq}f(p)f(q) q^md\omega d^3\!q d^3\!p.
\end{align*}
We combine the above expressions to write 
\begin{align}
&\int_{\bbr^3} Q_\varepsilon(f,f)(p)p^md^3\!p \nonumber\\
&= \frac{C_1}{2}\iiint_{\varrho\geq\varepsilon}\frac{\varrho^{2-b}}{pq}f(p)f(q)\Big((p')^m + (q')^m - p^m - q^m\Big)d\omega d^3\!q d^3\!p.\label{intQmod}
\end{align}
Let us consider the case $m=0$. In this case we obtain from \eqref{dsfm} and \eqref{intQmod} that
\[
\frac{d}{ds} \int_{\bbr^3} f(p)d^3p = 0.
\]
Since $f$ is non-negative, we can see that for all $s\geq 0$
\begin{align}\label{fL1}
\|f(s)\|_{L^1} = \|f_0\|_{L^1}.
\end{align}
In the case $m=2$ we apply Lemma \ref{Lem.Povzner} to \eqref{dsfm} and \eqref{intQmod} to obtain
\begin{align*}
\frac{d}{ds} \int_{\bbr^3} f(p)p^2d^3\!p &\leq C_1 \iiint_{\varrho\geq\varepsilon} \varrho^{2-b}f(p)f(q)d\omega d^3\!q d^3\!p\\
&\leq C\int_{\bbr^3}\int_{\bbr^3} (1+p^2q^2) f(p) f(q) d^3\!q d^3\!p\\
&\leq C \|f(s)\|^2_{L^1} +C\bigg( \int_{\bbr^3} f(p)p^2d^3\!p\bigg)^2,
\end{align*}
where we used Lemma \ref{Lem.varrho} and the fact that $0<2-b<1$. Note that the constant $C$ does not depend on $\varepsilon$. Then, applying \eqref{fL1} and Gr{\"o}nwall's inequality to the above we obtain that there exists $T>0$ such that
\begin{align}\label{fL12}
\sup_{0\leq s\leq T}\int_{\bbr^3} f(p)p^2 d^3\!p \leq C_T,
\end{align}
where $T$ and $C_T$ do not depend on $\varepsilon$.

In the case $m=-2$ we estimate the expression \eqref{intQmod} as
\begin{align*}
&\int_{\bbr^3} Q_\varepsilon(f,f)(p)\frac{1}{p^2}d^3\!p \nonumber\\
&\leq \frac{C_1}{2}\iiint_{\varrho\geq\varepsilon}\frac{\varrho^{2-b}}{pq}f(p)f(q)\bigg(\frac{1}{(p')^2} + \frac{1}{(q')^2} \bigg)d\omega d^3\!q d^3\!p
\end{align*}
and apply Lemma \ref{Lem.Povzner2} to obtain
\begin{align*}
\int_{\bbr^3} Q_\varepsilon(f,f)(p)\frac{1}{p^2}d^3\!p 
\leq C \int_{\bbr^3}\int_{\bbr^3}\frac{\varrho^{-b}}{pq}f(p)f(q) d^3\!q d^3\!p.
\end{align*}
Let us consider the integration over $\bq$ on the right hand side. By Lemma \ref{Lem.varrho2} we have
\begin{align*}
\int_{\bbr^3}\frac{\varrho^{-b}}{pq}f(q) d^3\!q 
& = 2^{-b}\int_{\bbr^3} \frac{1}{(pq)^{1+\frac{b}{2}}\sin^b(\theta/2)}f(q)d^3\!q\\
& = \frac{2^{1-b}\pi}{p^{1+\frac{b}{2}}}\int_{\bbr^+}\int_0^\pi \frac{q^2\sin\theta}{q^{1+\frac{b}{2}}\sin^b(\theta/2)}f(q)d\theta dq\\
& = \frac{2^{1-b}\pi}{p^{1+\frac{b}{2}}}\int_{\bbr^+} q^{1-\frac{b}{2}}f(q) dq \int_0^\pi \sin\theta \sin^{-b}(\theta/2) d\theta,
\end{align*}
where we used the fact that $f$ is isotropic. Note that\footnote{This is where we first need $b<2$.}
\begin{align*}
\int_0^\pi \sin\theta \sin^{-b}(\theta/2) d\theta
= \int_0^\pi 2\cos(\theta/2) \sin^{1-b}(\theta/2) d\theta
= \frac{4}{2-b}.
\end{align*}
Hence, we can estimate
\begin{align*}
\int_{\bbr^3}\frac{\varrho^{-b}}{pq}f(q) d^3\!q 
&\leq \frac{C}{p^{1+\frac{b}{2}}}\int_{\bbr^+} q^{1-\frac{b}{2}}f(q) dq\\
&\leq \frac{C}{p^{1+\frac{b}{2}}}\int_{\bbr^3} f(q)\frac{1}{q^{1+\frac{b}{2}}}d^3\!q.
\end{align*}
Considering the integration over $\bp$ we now obtain
\begin{align*}
\int_{\bbr^3} Q_\varepsilon(f,f)(p)\frac{1}{p^2}d^3\!p 
&\leq C \bigg(\int_{\bbr^3} f(p) \frac{1}{p^{1+\frac{b}{2}}}d^3\!p\bigg)^2\\
&\leq C\|f(s)\|_{L^1}^2+C\bigg(\int_{\bbr^3} f(p) \frac{1}{p^2}d^3\!p\bigg)^2.
\end{align*}
Applying \eqref{dsfm}, \eqref{fL1}, and Gr{\"o}nwall's inequality to the above, we obtain that there exists $T>0$ such that
\begin{align}\label{fL1-2}
\sup_{0\leq s\leq T}\int_{\bbr^3} f(p)\frac{1}{p^2}d^3\!p\leq C_T.
\end{align}
We combine the results \eqref{fL12} and \eqref{fL1-2} to complete the proof. 
\end{proof}

\subsection{Solutions to the unmodified Boltzmann equation}\label{ss3.2}
We need to remove the modification by letting $\varepsilon\to 0$. To be explicit, let $f_k$ denote the solution constructed in Proposition \ref{Prop.mod} with $\varepsilon = k^{-1}$ and $k=1,2,\cdots$. Then, it satisfies
\[
\partial_s f_k =  Q_k(f_k,f_k),\quad f_k(0) = f_0\geq 0,
\]
where
\[
Q_k(g,g) =  C_1 \int_{\bbr^3}\int_{\bbs^2} \mathds{1}_{\{\varrho\geq k^{-1}\}}\frac{\varrho^{2-b}}{pq}\Big(g(p')g(q')-g(p)g(q)\Big)d\omega d^3\!q,
\]
where $\mathds{1}_A$ is the usual characteristic function of $A$. 
The strategy is to show that $\{f_k\}$ is a Cauchy sequence in a suitable weighted $L^1$ space. 
Let us take $k<m$ and write $\partial_s f_k -\partial_s f_m$ as in \eqref{Qsymm}:
\begin{align*}
&\partial_s f_k -\partial_s f_m = Q_k(f_k,f_k) - Q_m(f_m,f_m)\allowdisplaybreaks\\
& = C_1\int_{\bbr^3}\int_{\bbs^2} \mathds{1}_{\{\varrho\geq k^{-1}\}}\frac{\varrho^{2-b}}{pq}\Big(f_k(p')f_k(q')-f_k(p)f_k(q)\\
& \qquad\qquad\qquad\qquad\qquad\qquad - f_m(p')f_m(q')+f_m(p)f_m(q)\Big)d\omega d^3\!q\\
&\quad - C_1\int_{\bbr^3}\int_{\bbs^2} \mathds{1}_{\{m^{-1}\leq\varrho\leq k^{-1}\}}\frac{\varrho^{2-b}}{pq}\Big(f_m(p')f_m(q')-f_m(p)f_m(q)\Big)d\omega d^3\!q\allowdisplaybreaks\\
& = \frac{C_1}{2}\int_{\bbr^3}\int_{\bbs^2} \mathds{1}_{\{\varrho\geq k^{-1}\}}\frac{\varrho^{2-b}}{pq}\\
&\qquad\times\Big( (f_k+f_m)(p')(f_k - f_m)(q') +(f_k +f_m)(q')(f_k - f_m)(p')\\
&\qquad\qquad - (f_k+f_m)(p)(f_k - f_m)(q) -(f_k +f_m)(q)(f_k - f_m)(p)\Big) d\omega d^3\!q\\
&\quad - C_1\int_{\bbr^3}\int_{\bbs^2} \mathds{1}_{\{m^{-1}\leq\varrho\leq k^{-1}\}}\frac{\varrho^{2-b}}{pq}\Big(f_m(p')f_m(q')-f_m(p)f_m(q)\Big)d\omega d^3\!q.
\end{align*}
Since $\partial_s |f_k - f_m | = \mbox{sgn}(f_k - f_m) \partial_s (f_k - f_m)$, we multiply the above expression by $\mbox{sgn}(f_k - f_m)$ to obtain
\begin{align}
&\partial_s| (f_k - f_m )(p)|\nonumber\\
& \leq \frac{C_1}{2}\int_{\bbr^3}\int_{\bbs^2} \frac{\varrho^{2-b}}{pq}\Big( (f_k+f_m)(p')|(f_k - f_m)(q')|+(f_k +f_m)(q')|(f_k - f_m)(p')|\nonumber\\
& \qquad\qquad\qquad\qquad+ (f_k+f_m)(p)|(f_k - f_m)(q)| \Big) d\omega d^3\!q\nonumber\\
&\quad + C_1\int_{\bbr^3}\int_{\bbs^2} \mathds{1}_{\{\varrho\leq k^{-1}\}}\frac{\varrho^{2-b}}{pq}\Big(f_m(p')f_m(q')+f_m(p)f_m(q)\Big)d\omega d^3\!q,\label{partial_f-f}
\end{align}
where we used $\mbox{sgn}(f_k - f_m)(p) (f_k - f_m)(p)=|(f_k - f_m)(p)|$ and the fact that the solutions are non-negative. 

Integrating \eqref{partial_f-f} over $\bp$ we obtain by \eqref{dpdq}
\begin{align*}
&\frac{d}{ds}\int_{\bbr^3}|(f_k-f_m)(p)|d^3\!p\\
& \leq  \frac{3C_1}{2}\iiint \frac{\varrho^{2-b}}{pq} (f_k+f_m)(p)|(f_k - f_m)(q)| d\omega d^3\!q d^3\!p\\
& \quad + C_1\iiint \mathds{1}_{\{\varrho\leq k^{-1}\}}\frac{\varrho^{2-b}}{pq}f_m(p)f_m(q)d\omega d^3\!q d^3\!p\\
& =: I_1 + I_2.
\end{align*}
The integral $I_1$ is estimated as follows:
\begin{align}
I_1&\leq C\iint \frac{1}{(pq)^{\frac{b}{2}}} (f_k+f_m)(p)|(f_k - f_m)(q)| d^3\!q d^3\!p\nonumber\\
&\leq C_T\int_{\bbr^3} |(f_k - f_m)(q)|\frac{1}{q^{\frac{b}{2}}} d^3\!q\nonumber\\
&\leq C_T\bigg(\int_{\bbr^3}|(f_k - f_m)(q)| d^3\!q + \int_{\bbr^3} |(f_k - f_m)(q)|\frac{1}{q} d^3\!q\bigg),\label{f-f-1}
\end{align}
where we used Lemma \ref{Lem.varrho} and the properties of solutions given in Proposition \ref{Prop.mod} with $1<b<2$. The integral $I_2$ is estimated as follows:
\begin{align}
I_2&\leq Ck^{-2+b}\iint \frac{1}{pq}f_m(p) f_m(q) d^3\!qd^3\!p\leq C_T k^{-2+b}. \label{f-f-2}
\end{align}
In a similar way, multiplying \eqref{partial_f-f} by $1/p$ and integrating it over $\bp$ we obtain
\begin{align*}
&\frac{d}{ds}\int_{\bbr^3}|(f_k - f_m)(p)|\frac{1}{p}d^3\!p\\
& \leq \frac{C_1}{2}\iiint  \frac{\varrho^{2-b}}{pq}\Big( (f_k+f_m)(p')|(f_k - f_m)(q')|+(f_k +f_m)(q')|(f_k - f_m)(p')|\nonumber\\
& \qquad\qquad\qquad\qquad+ (f_k+f_m)(p)|(f_k - f_m)(q)| \Big)\frac{1}{p} d\omega d^3\!q d^3\!p\nonumber\\
&\quad + C_1\iiint \mathds{1}_{\{\varrho\leq k^{-1}\}}\frac{\varrho^{2-b}}{pq}\Big(f_m(p')f_m(q')+f_m(p)f_m(q)\Big)\frac{1}{p}d\omega d^3\!qd^3\!p\allowdisplaybreaks\\
& = \frac{C_1}{2}\iiint  \frac{\varrho^{2-b}}{pq}(f_k+f_m)(p)|(f_k - f_m)(q)|\bigg(\frac{1}{p'}+\frac{1}{q'} +\frac{1}{p}\bigg) d\omega d^3\!q d^3\!p\\
&\quad + C_1\iiint \mathds{1}_{\{\varrho\leq k^{-1}\}}\frac{\varrho^{2-b}}{pq}f_m(p)f_m(q)\bigg(\frac{1}{p'}+\frac{1}{p}\bigg)d\omega d^3\!qd^3\!p,
\end{align*}
where we used the same argument as in \eqref{intQmod}. Let $J_1$, $J_2$, $J_3$, $J_4$, and $J_5$ denote the integrals in the last expression, for instance $J_1$ is the integral containing $1/p'$ in the first term, $J_2$ is the one containing $1/q'$ in the first term, etc. 

The integrals $J_3$ and $J_5$ are easily estimated as follows:
\begin{align*}
J_3 &= \frac{C_1}{2}\iiint  \frac{\varrho^{2-b}}{pq}(f_k+f_m)(p)|(f_k - f_m)(q)|\frac{1}{p} d\omega d^3\!q d^3\!p\\
&\leq C\int_{\bbr^3} (f_k + f_m)(p)\frac{1}{p^{1+\frac{b}{2}}}d^3\!p \int_{\bbr^3} |(f_k - f_m)(q)|\frac{1}{q^{\frac{b}{2}}} d^3\!q.
\end{align*}
Since $1<b<2$, we obtain
\begin{align}
J_3\leq C_T\bigg(\int_{\bbr^3}|(f_k - f_m)(q)| d^3\!q + \int_{\bbr^3} |(f_k - f_m)(q)|\frac{1}{q} d^3\!q\bigg).\label{f-f-3}
\end{align}
Similarly, we have
\begin{align}
J_5 & = C_1\iiint \mathds{1}_{\{\varrho\leq k^{-1}\}}\frac{\varrho^{2-b}}{pq}f_m(p)f_m(q)\frac{1}{p}d\omega d^3\!qd^3\!p\nonumber\\
&\leq C k^{-2+b}\int_{\bbr^3} f_m(p)\frac{1}{p^2}d^3\!p \int_{\bbr^3} f_m(q)\frac{1}{q} d^3\!q\nonumber\\
&\leq C_T k^{-2+b}.\label{f-f-4}
\end{align}
To estimate $J_1$, $J_2$, and $J_4$, we use Lemma \ref{Lem.Povzner1}. We first choose $0<a<1$ satisfying
\[
a+b<2,
\]
(which is possible since we have fixed $b<2$). Then, $J_1$ is estimated as follows:
\begin{align*}
J_1 & = \frac{C_1}{2}\iiint  \frac{\varrho^{2-b}}{pq}(f_k+f_m)(p)|(f_k - f_m)(q)|\frac{1}{p'} d\omega d^3\!q d^3\!p\\
& \leq C\iint \frac{\varrho^{2-b-a}}{pq(p+q)^{1-a}}(f_k+f_m)(p)|(f_k - f_m)(q)| d^3\!q d^3\!p\allowdisplaybreaks\\
& \leq C\iint \frac{1}{(pq)^{\frac{a+b}{2}}(p+q)^{1-a}}(f_k+f_m)(p)|(f_k - f_m)(q)| d^3\!q d^3\!p.
\end{align*}
Recall that Young's inequality can be written as $x+y \geq c x^{1-\gamma} y^{\gamma}$ for $0<\gamma<1$. To apply this inequality we choose $\gamma$ as
\[
0<\gamma\leq\frac{2-a-b}{2(1-a)}.
\]
Then, we observe that
\[
\frac{1}{(pq)^{\frac{a+b}{2}}(p+q)^{1-a}}\leq \frac{C}{p^{\frac{a+b}{2}+(1-\gamma)(1-a)}q^{\frac{a+b}{2}+\gamma(1-a)}},
\]
where $0<\frac{a+b}{2}+(1-\gamma)(1-a)< 2$ and $0<\frac{a+b}{2}+\gamma(1-a)\leq 1$ by the assumptions on $a$, $b$, and $\gamma$. Hence, $J_1$ can be estimated as 
\begin{align}
J_1&\leq C\iint \frac{1}{p^{\frac{a+b}{2}+(1-\gamma)(1-a)}q^{\frac{a+b}{2}+\gamma(1-a)}}(f_k+f_m)(p)|(f_k - f_m)(q)| d^3\!q d^3\!p\nonumber\\
&\leq C_T\int_{\bbr^3} \frac{1}{q^{\frac{a+b}{2}+\gamma(1-a)}}|(f_k - f_m)(q)| d^3\!q \nonumber\\
&\leq C_T\bigg(\int_{\bbr^3}|(f_k - f_m)(q)| d^3\!q + \int_{\bbr^3} |(f_k - f_m)(q)|\frac{1}{q} d^3\!q\bigg).\label{f-f-5}
\end{align}
The estimate of $J_2$ is exactly the same with $J_1$:
\begin{align}
J_2&\leq C_T\bigg(\int_{\bbr^3}|(f_k - f_m)(q)| d^3\!q + \int_{\bbr^3} |(f_k - f_m)(q)|\frac{1}{q} d^3\!q\bigg).\label{f-f-6}
\end{align}
With the exponents $a$ and $\gamma$ chosen as above the integral $J_4$ is estimated as follows:
\begin{align}
J_4&=C_1\iiint \mathds{1}_{\{\varrho\leq k^{-1}\}}\frac{\varrho^{2-b}}{pq}f_m(p)f_m(q)\frac{1}{p'} d\omega d^3\!qd^3\!p\nonumber\\
&\leq C\iint \mathds{1}_{\{\varrho\leq k^{-1}\}}\frac{\varrho^{2-b-a}}{pq(p+q)^{1-a}}f_m(p)f_m(q) d^3\!q d^3\!p\nonumber\allowdisplaybreaks\\
&\leq C k^{-2+a+b}\iint \frac{1}{p^{1+(1-\gamma)(1-a)}q^{1+\gamma(1-a)}}f_m(p)f_m(q) d^3\!q d^3\!p\nonumber\\
&\leq C k^{-2+a+b}\iint \frac{1}{p^{1+(1-\gamma)(1-a)}q^{1+\gamma(1-a)}}f_m(p)f_m(q) d^3\!q d^3\!p\nonumber\\
&\leq C_T k^{-2+a+b},\label{f-f-7}
\end{align}
where we used the fact that $1+(1-\gamma)(1-a)<2$ and $1+\gamma(1-a)<2$ and Proposition \ref{Prop.mod}. 

Let $X_{km}$ denote 
\[
X_{km}(s) = \int_{\bbr^3}|(f_k - f_m)(s,p)|\bigg(1+\frac{1}{p}\bigg)d^3\!p. 
\]
The estimates \eqref{f-f-1}--\eqref{f-f-7} show that $X_{km}$ satisfies the differential inequality:
\[
\frac{dX_{km}(s)}{ds}\leq C_Tk^{-2+a+b} +C_TX_{km}(s).
\]
Note that $f_k(0) = f_m(0) = f_0$ and $X_{km}(0) = 0$. Hence, we obtain on $[0,T]$,
\[
X_{km}(s)\leq C_T k^{-2+a+b}.
\]
Since $a+b<2$ and $k<m$, we conclude that $\{f_k\}$ converges in $L^1$ with weight $1+1/p$ as $k\to\infty$. We obtain the following result.
\addtocounter{thm}{-1}
\begin{thm}\label{thm1}
Let $f_0\in L^1(\bbr^3)$ be initial data satisfying $f_0\geq 0$ and 
\[
\int_{\bbr^3}f_0(p)p^m d^3\!p<\infty
\] 
for $-2\leq m\leq 2$. Then, there exists a positive value $T$ of the redefined time coordinate $s$ such that the Boltzmann equation \eqref{bolt5} has a unique solution $f\in C^1([0,T];L^1(\bbr^3))$ in $s$-time which is non-negative and satisfies
\[
\sup_{0\leq s\leq T}\int_{\bbr^3} f(s,p)p^m d^3\!p\leq C_T,
\]
but where now $-1\leq m\leq2$.
\end{thm}
Note that:
\begin{itemize}
\item The Gr\"onwall arguments in Section 3.1 show that the rescaling $f_0(p)\rightarrow\lambda f_0(p)$ has the effect $T\rightarrow\lambda^{-1}T$ on the time of existence while, in the case $\Lambda>0$, we compute the effect on the total life-time in $s$, which we can call
\[s_\infty:=\int_0^\infty R^{3-b}dt,\]
as $s_\infty\rightarrow\lambda^{(b-3)/4}s_\infty$; thus
\[\frac{T}{s_\infty}\rightarrow\lambda^{-(1+b)/4}\frac{T}{s_\infty},\]
and by choice of $\lambda$ (small) we can ensure that this is greater than one. That means that for $\Lambda>0$ and small data we have long-time existence in proper time $t$.
 \item The change of time coordinate from $t$ to $s$ eliminates the scale-factor $R(t)$ from the equations, so that this result gives finite-time existence in $s$ for Einstein-Boltzmann with all signs of $\Lambda$; this translates as finite-time in $t$ for $\Lambda\leq 0$ and finite or (as above) possibly infinite time in $t$ for $\Lambda>0$.
 
  \item By the same argument, the theorem gives finite-time existence in proper-time for massless Boltzmann with this family of cross-sections in Minkowski space.
 \item The argument also yields a continuation criterion: provided that
 \[
\int_{\bbr^3}f(s,p)p^m d^3\!p<\infty \mbox{   for  }m=-2
\] 
 continues to hold, since it is straightforward to see that this moment with $m=2$ grows only linearly in time, and the other relevant values of $m$ can then be controlled, the theorem allows a larger $T$.
\end{itemize}

{\bf{Acknowledgements:}}
In the course of this work the authors have benefitted with financial support from a number of institutions. PT acknowledges financial support from the Spanish Ministry of Economy and Competitiveness, through the “Severo Ochoa Programme for Centres of Excellence in R\&D” (SEV-2015-0554) in October 2017, and continuing support from St John's College, Oxford; HL, EN and PT all acknowledge support from the 
Erwin Schr\"odinger Institute, Vienna in February 2018; HL and EN acknowledge support from St John's College and the Oxford Mathematical Institute in March 2019; HL was supported by the Basic Science Research Program through the National Research Foundation of Korea (NRF) funded by the Ministry of Science, ICT \& Future Planning (NRF-2018R1A1A1A05078275); we also acknowledge useful conversations with Prof John Stalker of Trinity College Dublin.

\bibliographystyle{abbrv}

\end{document}